\documentclass{proc-l}

\issueinfo{140}
  {4}
  {April}
  {2012}
\PII{S \ISSN(06)08550-9}
\copyrightinfo{2011}{American Mathematical Society}
\pagespan{1321}{1330}
\commby{Walter Van Assche}
\date{April 1, 2009 and, in revised form, December 30, 2010}

\usepackage{hyperref}

\newcommand*{\mailto}[1]{\href{mailto:#1}{\nolinkurl{#1}}}

\newtheorem{theorem}{Theorem}[section]

\newtheorem{corollary}[theorem]{Corollary}

\newcommand{\R}{\mathbb{R}}
\newcommand{\N}{\mathbb{N}}
\newcommand{\Z}{\mathbb{Z}}
\newcommand{\C}{\mathbb{C}}

\newcommand{\nn}{\nonumber}
\newcommand{\be}{\begin{equation}}
\newcommand{\ee}{\end{equation}}
\newcommand{\bea}{\begin{eqnarray}}
\newcommand{\eea}{\end{eqnarray}}
\newcommand{\ul}{\underline}
\newcommand{\ol}{\overline}
\newcommand{\ti}{\tilde}

\newcommand{\spr}[2]{\langle #1 , #2 \rangle}

\newcommand{\I}{\mathrm{i}}

\newcommand{\tl}{\mathrm{TL}}
\newcommand{\AL}{\mathrm{AL}}
\newcommand{\km}{\mathrm{KM}}

\newcommand{\deven}{\delta_{\rm even}}
\newcommand{\dodd}{\delta_{\rm odd}}

\numberwithin{equation}{section}

\begin{document}

\title{Unique continuation\\ for discrete nonlinear wave equations}

\author[H. Kr\"uger]{Helge Kr\"uger}
\address{Department of Mathematics\\ Rice University\\ Houston\\ TX 77005\\ USA}
\curraddr{Department of Mathematics\\ Caltech\\ Pasadena\\CA 91125\\USA}
\email{\mailto{helge@caltech.edu}}
\urladdr{\url{http://www.its.caltech.edu/~helge/}}

\author[G. Teschl]{Gerald Teschl}
\address{Faculty of Mathematics\\ University of Vienna\\
Nordbergstrasse 15\\ 1090 Wien\\ Austria\\ and International Erwin Schr\"odinger
Institute for Mathematical Physics\\ Boltzmanngasse 9\\ 1090 Wien\\ Austria}
\email{\mailto{Gerald.Teschl@univie.ac.at}}
\urladdr{\url{http://www.mat.univie.ac.at/~gerald/}}

\thanks{Research supported by the Austrian Science Fund (FWF) under Grant No.\ Y330 and
the National Science Foundation (NSF) under Grant No.\ DMS--0800100.}


\begin{abstract}
We establish unique continuation for various discrete nonlinear wave equations.
For example, we show that if two solutions of the Toda lattice coincide for one lattice point in some
arbitrarily small time interval, then they coincide everywhere. Moreover, we establish analogous results
for the Toda, Kac--van Moerbeke, and Ablowitz--Ladik hierarchies. Although all these equations
are integrable, the proof does not use integrability and can be adapted to other equations as well.
\end{abstract}

\maketitle

\section{Introduction}

Unique continuation results for wave equations have a long tradition and seem to originate in control theory.
One of the first results seems to be the one by Zhang \cite{zh}, where he proves that if a short-range solution
of the Korteweg--de Vries (KdV) equation vanishes on an open subset in the $x/t$-plane, then it must vanish everywhere.
Since then, this result has been extended in various directions and for different equations (see for example \cite{bu},
the introduction in \cite{ik} for the case of the nonlinear Schr\"odinger equation, \cite{ekpv}, \cite{kpv}, \cite{kpv2} for the
generalized KdV equation, \cite{le} for the Camassa--Holm equation).

However, all the results so far seem to only deal with wave equations which are continuous in the spatial direction and
this clearly raises the question for such unique continuation results for wave equations which are discrete in the spatial
variable. In particular, to the best of our knowledge, there are no results for example for the Toda equation,
one of the most prominent discrete systems. While in principle the strategy from Zhang \cite{zh} would be applicable
to the Toda lattice, it is the purpose of this paper to advocate a much simpler direct approach in the
discrete case. We will start with the Toda lattice as our prototypical example and then show how the entire Toda
hierarchy as well as the  Kac--van Moerbeke and Ablowitz--Ladik hierarchies can be treated. It is important to
stress that our approach does not use integrability of these equations and hence can be adapted to more
general systems. On the other hand, our approach is restricted to one dimension in the spatial variable and
thus does not apply to the discrete Schr\"odinger equation on $\Z^d$. Due to the connections with
localization for discrete Anderson--Bernoulli models, unique continuation for this
model is an important open problem; see \cite{bu2}, \cite{buk}.

\section{The Toda lattice}

In this section we want to treat the Toda lattice as the prototypical example.
To this end, recall the Toda lattice \cite{ta} (in Flaschka's variables \cite{fl1})
\begin{align} \nn
\dot{a}(n,t) &= a(n,t) \Big(b(n+1,t)-b(n,t)\Big), \\ \label{todeqfl}
\dot{b}(n,t) &= 2 \Big(a(n,t)^2-a(n-1,t)^2\Big), \qquad n\in \Z,
\end{align}
where the dot denotes a derivative with respect to $t$.
It is a well-studied physical model and one of the prototypical discrete integrable wave equations.
We refer to the monographs \cite{fad}, \cite{tjac}, \cite{ta} or the review articles \cite{krt}, \cite{taet}
for further information.

\begin{theorem}\label{thmToda}
Assume that $a_0(n,t), b_0(n,t)$ and $a(n,t), b(n,t)$ are complex-valued solutions
of the Toda lattice \eqref{todeqfl} with $a_0(n,t) \ne 0$ for all $(n,t) \in \Z\times\R$ such that there is one $n_0 \in \Z$
and two times $t_0<t_1$ such that
\be\label{eq:asN}
a_0(n_0,t)^2 = a(n_0,t)^2,\quad
b_0(n_0,t) = b(n_0,t),
\ee
for $t\in (t_0,t_1)$. Then
\be
a_0(n,t)^2 = a(n,t)^2,\quad b_0(n,t) = b(n,t)
\ee
for all $(n,t) \in \Z\times\R$.
\end{theorem}

\begin{proof}
It suffices to prove that \eqref{eq:asN} for $n_0$ implies
\eqref{eq:asN} for $n_0 - 1$ and $n_0 + 1$. We start with $N_0-1$ and first observe that \eqref{todeqfl} implies that
\begin{align*}
0 &= \dot b(n_0,t) - \dot b_0(n_0,t) = 2 \big(a(n_0,t)^2 - a_0(n_0,t)^2 - a(n_0-1,t)^2 + a_0(n_0-1,t)^2\big)\\
&= -2 \big(a(n_0-1,t)^2 - a_0(n_0-1,t)^2\big)
\end{align*}
and thus $a(n_0-1,t)^2 = a_0(n_0-1,t)^2$. Using this we compute
\begin{align*}
0 &=\frac{\dot a(n_0-1,t)}{a(n_0-1,t)} - \frac{\dot a_0(n_0-1,t)}{a_0(n_0-1,t)}\\
&= b(n_0,t) - b_0(n_0,t) - b(n_0-1,t) + b_0(n_0-1,t)\\
&= - b(n_0-1,t) + b_0(n_0-1,t),
\end{align*}
so $b(n_0-1,t) = b_0(n_0-1,t)$. Now for $n_0 + 1$, we begin with
\begin{align*}
0 &= \frac{\dot a(n_0,t)}{a(n_0,t)} - \frac{\dot a_0(n_0,t)}{a_0(n_0,t)} = b(n_0+1,t) - b_0(n_0+1,t) - b(n_0,t) + b_0(n_0,t)\\
&= b(n_0+1,t) - b_0(n_0+1,t),
\end{align*}
so $b(n_0+1,t) = b_0(n_0+1,t)$. Now, use that
\begin{align*}
0 &= \dot b(n_0+1,t) - \dot b_0(n_0+1,t)\\
&= 2 \big(a(n_0+1,t)^2 - a_0(n_0+1,t)^2 - a(n_0,t)^2 + a_0(n_0,t)^2\big)\\
&=2 \big(a(n_0+1,t)^2 - a_0(n_0+1,t)^2\big)
\end{align*}
to conclude that $a(n_0+1,t)^2 = a_0(n_0+1,t)^2$. This finishes the proof.
\end{proof}

It is worthwhile to note that the assumption $a_0(n,t)\ne 0$ is crucial. In fact, if $a_0(n_0,t)=0$ for one
(and hence for all) $t\in\R$, then the Toda lattice decouples into two independent
parts to the left and right of $n_0$, and the above result is clearly wrong. However, it remains valid
on every consecutive number of points for which $a_0(n,t)\ne 0$ holds true. In particular,
our result applies to the half-line Toda lattice or to the finite Toda lattice.

As a simple consequence, this also proves that the propagation speed for the Toda lattice is finite.

\begin{corollary}
Let $a(n,t)\ne 0$, $b(n,t)$ be a complex-valued solution of the Toda lattice \eqref{todeqfl} for which $a(n,t_0)-\frac{1}{2}$, $b(n,t_0)$
is supported on a finite number of points $n$ at some initial time $t_0$. Then this does not remain true for $t\in(t_0,t_1)$
unless $a(n,t)=\frac{1}{2}$, $b(n,t)=0$ for all $(n,t) \in \Z\times\R$.
\end{corollary}

In fact, in the case of real-valued solutions, one can even show the somewhat stronger result that
$a(n,t_0)-\frac{1}{2}$, $b(n,t_0)$ can be compactly supported for at most one time \cite{ttd}. However,
on the other hand, the Toda lattice does preserve certain asymptotic properties of the initial conditions;
see again \cite{ttd}.

\section{Extension to the Toda and Kac--van Moerbeke hierarchy}
\label{secth}

In this section we show that our main result extends to the entire Toda hierarchy (which will
cover the Kac--van Moerbeke hierarchy as well).
To this end, we introduce the Toda hierarchy using the standard Lax formalism
following \cite{bght} (see also \cite{ghmt}, \cite{tjac}).

Associated with two sequences $a(t)^2\ne 0, b(t)$ is a Jacobi operator
\begin{equation} \label{defjac}
H(t) = a(t) S^+  + a^-(t) S^-  + b(t)
\end{equation}
acting on sequences over $\Z$, where $S^\pm f(n) = f^\pm(n)= f(n\pm1)$ are the usual shift operators.
Moreover, choose constants $c_0=1$, $c_j$, $1\le j \le r$, $c_{r+1}=0$, and set
\be
P_{2r+2}(t) = \sum_{j=0}^r c_{r-j} \ti{P}_{2j+2}(t), \qquad
\ti{P}_{2j+2}(t) = [H(t)^{j+1}]_+ - [H(t)^{j+1}]_-,
\ee
where $[A]_\pm$ denote the upper and lower triangular parts of an operator with respect to the
standard basis $\delta_m(n)= \delta_{m,n}$ (with $\delta_{m,n}$ the usual Kronecker delta).
Then the Toda hierarchy is equivalent to the Lax equation
\begin{equation}
\frac{d}{dt} H(t) -[P_{2r+2}(t), H(t)]=0, \qquad t\in\R,
\end{equation}
where $[A,B]= A B - B A$ is the usual commutator. Abbreviating
\begin{align}\nn
g_j(n,t) &= \sum_{\ell=0}^j c_{j-\ell} \ti{g}_\ell(n,t), \quad \ti{g}_\ell(n,t)=\spr{\delta_n}{H(t)^\ell \delta_n},\\
h_j(n,t) &= \sum_{\ell=0}^j c_{j-\ell}  \ti{h}_\ell(n,t) + c_{j+1}, \quad \ti{h}_\ell(n,t) = 2 a(n,t) \spr{\delta_{n+1}}{H(t)^\ell\delta_n},
\end{align}
one explicitly obtains
\begin{equation}\label{tlrabo}
\tl_r (a(t), b(t)) = \begin{pmatrix} \dot{a}(t) -a(t) \Big(g_{r+1}^+(t) -
g_{r+1}(t) \Big)\\ 
\dot{b}(t) - \Big(h_{r+1}(t) -h_{r+1}^-(t) \Big) \end{pmatrix} =0, \qquad r\in\N_0,
\end{equation}
for the $r$-th equation $\tl_r(a,b) =0$ in the Toda hierarchy (where $\N_0=\N\cup\{0\}$).

Our main point in this section is the following generalization of Theorem~\ref{thmToda} to the entire Toda hierarchy:

\begin{theorem}\label{thmtlh}
Assume that $a_0(n,t)\ne 0, b_0(n,t)$ and $a(n,t), b(n,t)$ are complex-valued solutions
of some equation in the Toda hierarchy $\tl_r$ such that there is one $n_0 \in \Z$ and two times $t_0<t_1$
such that
\be
a_0(n_0+j,t)^2 = a(n_0+j,t)^2,\quad
b_0(n_0+j,t) = b(n_0+j,t), \qquad j=0,\dots, r,
\ee
for $t\in (t_0,t_1)$. Then
\be\label{a0b0eqab}
a_0(n,t)^2 = a(n,t)^2,\quad b_0(n,t) = b(n,t)
\ee
for all $(n,t) \in \Z\times\R$.
\end{theorem}

\begin{proof}
Let us drop the dependence on $t$ for notational simplicity during this proof.
The key observation is the following structure for the homogenous quantities
$\ti{g}_j$, $\ti{h}_j$:
\[
\ti{g}_j(n) = \begin{cases}
\left(\prod\limits_{\ell=0}^{k-1} a(n+\ell)^2 \right) b(n+k) + R(n+k-1,n-k+1) +\\
+ \left(\prod\limits_{\ell=1}^k a(n-\ell)^2 \right) \left( b(n-k) + 2 \sum\limits_{\ell=0}^{k-1} b(n-\ell)\right), \qquad j=2k+1,\\
 \left(\prod\limits_{\ell=0}^{k-2} a(n+\ell)^2 \right) \Big( a(n+k-1)^2 +  b(n+k-1)^2\\ + 2b(n+k-1)\sum\limits_{\ell=0}^{k-2} b(n+\ell) \Big)+ \\
+ R(n+k-2,n-k+1) + \prod\limits_{\ell=1}^k a(n-\ell)^2, \qquad j=2k,
\end{cases}
\]
and
\[
\ti{h}_j(n) = \begin{cases}
2 \left(\prod\limits_{\ell=0}^{k-1} a(n+\ell)^2 \right) \Big( a(n+k)^2 +  b(n+k)^2\\ + 2b(n+k)\sum\limits_{\ell=0}^{k-1} b(n+\ell) \Big)+\\
+ R(n+k-1,n-k+1) + 2\prod\limits_{\ell=0}^k a(n-\ell)^2, \qquad j=2k+1,\\
2\left(\prod\limits_{\ell=0}^{k-1} a(n+\ell)^2 \right) b(n+k) + R(n+k-1,n-k+2)\\
+ 2 \prod\limits_{\ell=0}^{k-1} a(n-\ell)^2 \left( b(n+1) + b(n-k+1) +2 \sum\limits_{\ell=0}^{k-2} b(n-\ell) \right), \qquad j=2k,
\end{cases}
\]
for $j>1$. Here $R(n,m)$ denotes terms which involve only $a(\ell)$ and $b(\ell)$ with $m \le \ell \le n$ and
we set $R(n,m)=0$ if $n<m$.
In fact, this can be verified using $\ti{g}_0=1$, $\ti{h}_0=0$, together with the recursions (\cite[Chap.~6]{tjac})
\begin{align}
\ti{g}_{j+1}  &= \frac{\ti{h}_j + \ti{h}_j^-}{2} + b \ti{g}_j,\\
\ti{h}_{j+1} &= 2a^2 \sum_{\ell=0}^j \ti{g}_{j-\ell} \ti{g}_{\ell}^+ - \frac{1}{2} 
\sum_{\ell=0}^j \ti{h}_{j-\ell} \ti{h}_{\ell}, \quad j\in\N_0.
\end{align}
Now we are ready for the main part of the proof. It suffices to show that \eqref{a0b0eqab} holds for $n=n_0-1$ and $n=n_0+r+1$.

We first look at the case $r+1=2k+1$. Then
\begin{align*}
0 &= \frac{\dot{a}(n_0+k)}{a(n_0+k)} - \frac{\dot{a}_0(n_0+k)}{a_0(n_0+k)}\\
&= g_{r+1}(n_0+k+1) - g_{r+1}(n_0+k) - g_{0,r+1}(n_0+k+1) + g_{0,r+1}(n_0+k)\\
&= \left( \prod_{\ell=k+1}^r a_0(n_0+\ell)^2 \right)\big( b(n_0+r+1) - b_0(n_0+r+1) \big)
\end{align*}
shows that $b(n_0+r+1) = b_0(n_0+r+1)$. Similarly,
\begin{align*}
0 &= \dot{b}(n_0+k) - \dot{b}_0(n_0+k)\\
&= h_{r+1}(n_0+k) - h_{r+1}(n_0+k-1) - h_{0,r+1}(n_0+k) + h_{0,r+1}(n_0+k-1)\\
&= 2 \left( \prod_{\ell=0}^{k-1} a_0(n_0+\ell)^2 \right)\big( a(n_0-1)^2 - a_0(n_0-1)^2 \big)
\end{align*}
shows that $a(n_0-1)^2 = a_0(n_0-1)^2$. Proceeding like this and using the result found in the previous steps,
\begin{align*}
0 &= \dot{b}(n_0+k+1) - \dot{b}_0(n_0+k+1)\\
&= h_{r+1}(n_0+k+1) - h_{r+1}(n_0+k) - h_{0,r+1}(n_0+k+1) + h_{0,r+1}(n_0+k)\\
&= \left(\prod\limits_{\ell=k+1}^r a_0(n_0+\ell)^2 \right) \big( a(n_0+r+1)^2 - a_0(n_0+r+1)^2\big)
\end{align*}
shows that $a(n_0+r+1)^2 = a_0(n_0+r+1)^2$, and
\begin{align*}
0 &= \frac{\dot{a}(n_0+k-1)}{a(n_0+k-1)} - \frac{\dot{a}_0(n_0+k-1)}{a_0(n_0+k-1)}\\
&= g_{r+1}(n_0+k) - g_{r+1}(n_0+k-1) - g_{0,r+1}(n_0+k) + g_{0,r+1}(n_0+k-1)\\
&= \left(\prod\limits_{\ell=-1}^k a_0(n_0+\ell)^2 \right) \big( b(n_0-1) - b_0(n_0-1) \big)
\end{align*}
shows that $b(n_0-1) = b_0(n_0-1)$, which finishes the case $r+1=2k+1$.
The case $r+1=2k$ is analogous.
\end{proof}

Finally, since the Kac--van Moerbeke hierarchy can be obtained by setting $b=0$ in the odd equations of the
Toda hierarchy, $\km_r(a) = \tl_{2r+1}(a,0)$ (see \cite{mt}), this last result also covers the Kac--van Moerbeke hierarchy.
In particular,

\begin{corollary}
Assume that $\rho_0(n,t)\ne 0$ and $\rho(n,t)\ne 0$ are solutions
of the Kac--van Moerbeke equation
\be
\dot{\rho}(n,t) = \rho(n,t) \big(\rho(n+1,t) - \rho(n-1,t) \big)
\ee
such that there is one $n_0 \in \Z$ and two times $t_0<t_1$
such that
\be
\rho_0(n_0,t) = \rho(n_0,t),\quad
\rho_0(n_0+1,t) = \rho(n_0+1,t),
\ee
for $t\in (t_0,t_1)$. Then
\be
\rho_0(n,t) = \rho(n,t)
\ee
for all $(n,t) \in \Z\times\R$.
\end{corollary}

\section{The Ablowitz--Ladik hierarchy}

In this section we show that our main result extends to the Ablowitz--Ladik (AL) hierarchy \cite{ghmt}.
We first state the result for the simplest case, whose proof follows
as the one of Theorem~\ref{thmToda}.

\begin{theorem}
Let $C_{0,\pm}, c_1\in\C\backslash \{0\}$.
Assume that $\alpha_0(n,t), \beta_0(n,t)$, with $\rho_0(n,t) \ne 0$, and
$\alpha(n,t), \beta(n,t)$ are solutions
of the Ablowitz--Ladik equation
\begin{align}\nn
\I \dot{\alpha}(n,t) &= -\rho(n,t)^2\big(c_{0,-} \alpha(n-1,t)+ c_{0,+} \alpha(n+1,t)\big) - c_1 \alpha(n,t),\\
\I \dot{\beta}(n,t) &= \rho(n,t)^2\big(c_{0,+} \beta(n-1,t)+ c_{0,-}\beta(n+1,t)\big) + c_1 \beta(n,t),
\end{align}
where
\be
\rho(n,t) = (1-\alpha(n,t)\beta(n,t))^{1/2},
\ee
such that there is one $n_0 \in \Z$ and two times $t_0<t_1$
such that
\be
\alpha_0(n_0+j,t) = \alpha(n_0+j,t),\quad
\beta_0(n_0+j,t) = \beta(n_0+j,t), \qquad j=0,1,
\ee
for $t\in (t_0,t_1)$. Then
\be
\alpha_0(n,t) = \alpha(n,t), \qquad \beta_0(n,t) = \beta(n,t)
\ee
for all $(n,t) \in \Z\times\R$.
\end{theorem}

The special choices $c_{0,\pm}=1$, $c_1=-2$, and $\beta= \pm \ol{\alpha}$ yield the focusing, defocusing
discrete nonlinear Schr\"odinger equations, respectively. The alternative choice $c_{0,\pm}=\pm\I$, $c_1=0$,
and $\beta=\ol{\alpha}$ yield the Schur flow.

We next turn to the AL hierarchy following \cite{ghmtalh}, \cite{ghmt}.
Associated with two sequences $\alpha(t), \beta(t)$ is a CMV operator
\begin{align}\nn
L(t) &= \rho^-(t) \rho(t) \deven S^{--} + (\beta^-(t)\rho(t) \deven - \alpha^+(t)\rho(t) \dodd) S^- 
- \beta(t)\alpha^+(t) \\
& \quad + (\beta(t) \rho^+(t) \deven - \alpha^{++}(t) \rho^+(t) \dodd) S^+ 
+ \rho^+(t) \rho^{++}(t) \dodd S^{++}, 
\end{align}
acting on sequences over $\Z$, where $\deven$ and $\dodd$ denote the characteristic functions of the even, odd integers,
\begin{equation}
\deven = \chi_{_{2\Z}}, \quad \dodd = 1 - \deven = \chi_{_{2\Z +1}},
\end{equation}
respectively. Next, consider
\begin{align} \nn
P_{\ul p}(t) =& \frac{\I}{2} \sum_{\ell=1}^{p_+} c_{p_+ -\ell,+} \big( [L^\ell(t)]_+ - [L^\ell(t)]_- \big)\\
& - \frac{\I}{2} \sum_{\ell=1}^{p_-} c_{p_- -\ell,-} \big( [L^{-\ell}(t)]_+ - [L^{-\ell}(t)]_- \big)  - \frac{\I}{2} c_{\ul p} \, Q_d,
\qquad \ul p\in\N_0^2,
\end{align} 
with $Q_d$ denoting the doubly infinite diagonal matrix
\begin{equation} 
Q_d=\big((-1)^k \delta_{k,\ell} \big)_{k,\ell \in\Z}.
\end{equation}
Then the AL hierarchy is equivalent to the Lax equation
\begin{equation}
\frac{d}{dt} L(t) -[P_{\ul p}(t), L(t)]=0, \qquad t\in\R.
\end{equation}
To find an explicit expression we introduce
\begin{align}\nn
f_{\ell, \pm}(t) &= \sum_{k=0}^\ell c_{\ell-k, \pm} \hat f_{k, \pm}(t), \quad
g_{\ell, \pm}(t) = \sum_{k=0}^\ell c_{\ell-k, \pm} \hat g_{k, \pm}(t), \\
h_{\ell, \pm}(t) &= \sum_{k=0}^\ell c_{\ell-k, \pm} \hat h_{k, \pm}(t),
\end{align}
where
\begin{align}
\hat f_{\ell,+}(n,t) &=  \alpha(n,t) \spr{\delta_n}{L^{ \ell+1} \delta_n}  + 
\rho(n,t)\begin{cases}  
\spr{\delta_{n-1}}{L^{ \ell+1}(t) \delta_n} , & n \text{ even,}\\
\spr{\delta_n}{L^{ \ell+1}(t) \delta_{n-1}},  & n \text{ odd,} \end{cases} \nn \\
\hat f_{\ell,-}(n,t) &= \alpha(n,t) (\delta_n,L^{- \ell} \delta_n)  + 
 \rho(n,t)\begin{cases}  
\spr{\delta_{n-1}}{L^{- \ell}(t) \delta_n},  & n \text{ even,}\\
\spr{\delta_n}{L^{- \ell}(t) \delta_{n-1}},  & n \text{ odd,} \end{cases}  \nn \\
\hat g_{0,\pm} &= 1/2, \quad 
\hat g_{\ell,\pm}(n,t) = \spr{\delta_n}{L^{\pm \ell}(t) \delta_n},  \\
\hat h_{\ell,+}(n,t) &= \beta(n,t) \spr{\delta_n}{L^{ \ell}(t) \delta_n}  + 
\rho(n,t)\begin{cases}  
\spr{\delta_n}{L^{ \ell}(t) \delta_{n-1}},  & n \text{ even,}\\
\spr{\delta_{n-1}}{L^{ \ell}(t) \delta_n} ,  & n \text{ odd,} \end{cases}\nn  \\
\hat h_{\ell,-}(n,t) &= \beta(n,t) \spr{\delta_n}{L^{- \ell-1} \delta_n}  + 
 \rho(n,t)\begin{cases}  
\spr{\delta_n}{L^{- \ell-1}(t) \delta_{n-1}} ,  & n \text{ even,}\\
\spr{\delta_{n-1}}{L^{- \ell-1}(t) \delta_n} ,  & n \text{ odd.} \end{cases} \nn
\end{align}
Then the $\ul p$th equation, $\ul p=(p_-,p_+)\in\N_0^2$, in the AL
hierarchy is given by
\begin{align}   
\begin{split}
& \AL_{\ul p} (\alpha, \beta) = \begin{pmatrix}-\I \dot{\alpha}(t)
- \alpha(g_{p_+,+}(t) + g_{p_-,-}^-(t)) + f_{p_+ -1,+}(t) - f_{p_- -1,-}^-(t)\\
-\I\dot{\beta}(t)+ \beta(g_{p_+,+}^-(t) + g_{p_-,-}(t)) 
- h_{p_- -1,-}(t) + h_{p_+ -1,+}^-(t) \end{pmatrix}=0,  \\
& \hspace*{7cm} \ul p=(p_-,p_+)\in\N_0^2.
\end{split}
\end{align}

\begin{theorem}
Fix some $\ul{p}=(p_+,p_-)\in \N_0^2$ such that $p_-=p_+>0$ and set $p=p_++p_--1$.
Assume that $\alpha_0(n,t), \beta_0(n,t)$, with $\rho_0(n,t) \ne 0$, and $\alpha(n,t), \beta(n,t)$
are solutions of some equation in the Toda hierarchy $\AL_{\ul p}$ such that there is one
$n_0 \in \Z$ and two times $t_0<t_1$
such that
\be
\alpha_0(n_0+j,t) = \alpha(n_0+j,t),\quad
\beta_0(n_0+j,t) = \beta(n_0+j,t), \qquad 0 \le j\le p,
\ee
for $t\in (t_0,t_1)$. Then
\be
\alpha_0(n,t) = \alpha(n,t),\quad \beta_0(n,t) = \beta(n,t)
\ee
for all $(n,t) \in \Z\times\R$.
\end{theorem}

\begin{proof}
Again we drop the dependence on $t$ for notational simplicity during this proof and
use the same conventions as in the proof of Theorem~\ref{thmtlh}.

The homogeneous quantities $\hat f_{\ell,\pm}$, $\hat g_{\ell,\pm}$, $\hat h_{\ell,\pm}$
are uniquely defined by the following recursion relations \cite[Lem.~C.5]{ghmt}:
\begin{align*}
\hat g_{0,+} &= \frac{1}{2}, \quad \hat f_{0,+} = -\alpha^+, \quad \hat h_{0,+} = \beta, \\
\hat g_{l+1,+} &= \sum_{k=0}^l \hat f_{l-k,+} \hat h_{k,+}
- \sum_{k=1}^l \hat g_{l+1-k,+} \hat g_{k,+}, \\
\hat f_{l+1,+}^- &= \hat f_{l,+} - \alpha (\hat g_{l+1,+} + \hat g_{l+1,+}^-),\\
\hat h_{l+1,+} &= \hat h_{l,+}^- + \beta (\hat g_{l+1,+} + \hat g_{l+1,+}^-), 
\end{align*}
and
\begin{align*}
\hat g_{0,-} &= \frac{1}{2}, \quad \hat f_{0,-} = \alpha, \quad \hat h_{0,-} = -\beta^+, \\
\hat g_{l+1,-} &= \sum_{k=0}^l \hat f_{l-k,-} \hat h_{k,-}
- \sum_{k=1}^l \hat g_{l+1-k,-} \hat g_{k,-}, \\
\hat f_{l+1,-} &= \hat f_{l,-}^- + \alpha (\hat g_{l+1,-} + \hat g_{l+1,-}^-),\\
\hat h_{l+1,-}^- &= \hat h_{l,-} - \beta (\hat g_{l+1,-} + \hat g_{l+1,-}^-). 
\end{align*}
From them we obtain
\begin{align}\nn
\hat f_{j,+}(n) =& -\left(\prod_{l=1}^j \rho(n+l)^2\right) \alpha(n+j+1) + R(n+j,n-j+2)\\
& +  \left(\prod_{l=0}^{j-2} \rho(n-l)^2\right) \alpha(n+1)^2 \beta(n-j+1),\\ \nn
\hat f_{j,-}(n) =& -\left(\prod_{l=1}^{j-1} \rho(n+l)^2\right) \alpha(n)^2 \beta(n+j) + R(n+j-1,n-j+1)\\
& + \left(\prod_{l=0}^{j-1} \rho(n-l)^2\right) \alpha(n-j),
\end{align}
\begin{align}\nn
\hat g_{j,+}(n) =& -\left(\prod_{l=1}^{j-1} \rho(n+l)^2\right) \beta(n)\alpha(n+j) + R(n+j-1,n-j+2)\\
& - \left(\prod_{l=0}^{j-2} \rho(n-l)^2\right) \alpha(n+1) \beta(n-j+1),\\ \nn
\hat g_{j,-}(n) =& -\left(\prod_{l=1}^{j-1} \rho(n+l)^2\right) \alpha(n)\beta(n+j) + R(n+j-1,n-j+2)\\
& - \left(\prod_{l=0}^{j-2} \rho(n-l)^2\right) \beta(n+1) \alpha(n-j+1),
\end{align}
\begin{align}\nn
\hat h_{j,+}(n) =& -\left(\prod_{l=1}^{j-1} \rho(n+l)^2\right) \beta(n)^2 \alpha(n+j) + R(n+j-1,n-j+1)\\
& + \left(\prod_{l=0}^{j-1} \rho(n-l)^2\right) \beta(n-j),\\ \nn
\hat h_{j,-}(n) =& -\left(\prod_{l=1}^j \rho(n+l)^2\right) \beta(n+j+1) + R(n+j,n-j+2)\\
& +  \left(\prod_{l=0}^{j-2} \rho(n-l)^2\right) \beta(n+1)^2 \alpha(n-j+1)
\end{align}
for $j\in\N$.
Note that it suffices to verify the $+$ case since the $-$ case follows from $\hat f_{j,\pm}(\alpha,\beta)= 
\hat h_{j,\pm}(\alpha,\beta)$ and $\hat g_{j,+}(\alpha,\beta)=  \hat g_{j,\-}(\alpha,\beta)$ (\cite[Lem.~3.7]{ghmt}).

Now we can proceed as in the case of the Toda hierarchy. For example,
\[
0= \I \big( \dot{\alpha}(n+p_-) - \dot{\alpha}_0(n+p_-) \big)= - c_{0,+} \left(\prod_{l=p_-}^{p} \rho(n-l)^2\right)
\big( \alpha(n+p+1) -  \alpha_0(n+p+1)\big)
\]
implies that $\alpha(n+p+1) = \alpha_0(n+p+1)$, etc.
\end{proof}

Interestingly, the above approach does not seem to work for $p_-\ne p_+$ in general. In any case,
the above result covers the discrete nonlinear Schr\"odinger and Schur hierarchies via the
above-mentioned special choices $c_{0,\pm}=1$, $\beta= \pm \ol{\alpha}$ and
$c_{0,\pm}=\pm\I$, $\beta=\ol{\alpha}$.

{\bf Acknowledgments.}
We thank F.\ Gesztesy and the anonymous referee for pointing out errors in a previous version of this article.

\end{document}